%% file: paperLong.tex
\documentclass{article}
\usepackage{spconf,amsmath,epsfig}

\usepackage{amsmath}
\usepackage{amssymb}
\usepackage{amsthm}
\usepackage{algorithm}
\usepackage{algpseudocode}
\usepackage{graphicx}
\usepackage{bbold}
\usepackage{newclude}
\usepackage{multicol}
\usepackage{multirow}
\usepackage{calc}

\include*{shortcutsNiclas}

\makeatletter
\newcommand*\@dblLabelI {}
\newcommand*\@dblLabelII {}
\newcommand*\@dblequationAux {}
\def\@dblequationAux #1,#2,%
    {\def\@dblLabelI{\label{#1}}\def\@dblLabelII{\label{#2}}}
\newcommand*{\doubleequation}[3][]{%
    \par\vskip\abovedisplayskip\noindent
    \if\relax\detokenize{#1}\relax
       \let\@dblLabelI\@empty
       \let\@dblLabelII\@empty
    \else 
       \@dblequationAux #1,%
    \fi
    \makebox[0.5\linewidth-1.5em]{%
     \hspace{\stretch2}%
     \makebox[0pt]{$\displaystyle #2$}%
     \hspace{\stretch1}%
    }%
    \makebox[0.5\linewidth-1.5em]{%
     \hspace{\stretch1}%
     \makebox[0pt]{$\displaystyle #3$}%
     \hspace{\stretch2}%
    }%
    \makebox[3em][r]{(%
  \refstepcounter{equation}\theequation\@dblLabelI, 
  \refstepcounter{equation}\theequation\@dblLabelII)}%
  \par\vskip\belowdisplayskip
}
\makeatother

\newcommand{\x}{x^\text{opt}_\lambda}
\newcommand{\xs}{x^\text{opt}_{\lambda^\star}}
\newcommand{\xsi}{x^\text{opt}_{\lambda^\star_i}}
\newcommand{\xtwo}{x^{\text{opt,rlx}}_\lambda}

\newcommand{\aperp}{a_{\lambda^\star}\left(W^\perp\right)}
\newcommand{\asW}{a_{\lambda^\star}(W)}

\newcommand{\hs}{h_{\lambda^\star}}
\newcommand{\hperp}{h_{\lambda^\star}\left(W^\perp\right)}
\newcommand{\hsW}{h_{\lambda^\star}(W)}

\newcommand{\ls}{\lambda^\star}

\newcommand{\Monegen}{M_\text{Alg1}}
\newcommand{\Mone}{M_\text{Alg1}^{W=0}}
\newcommand{\Mtwo}{M_\text{Alg2}}

\newcommand{\dg}{d_{\lambda^\star}(\lambda,W)}
\newcommand{\sg}{s_{\ls}(\lambda)}

\newcommand{\Cperp}{\mathcal{C}_{\lambda^\star}^\perp(\lambda)}

\title{Approximate Regularization Paths for Nuclear Norm Minimization Using Singular Value Bounds -- With Implementation and Extended Appendix}
\name{N. Blomberg$^\dagger$, C.R. Rojas$^\dagger$, B. Wahlberg$^\dagger$ \thanks{This work was supported by the European Research Council under the advanced grant
LEARN, contract 267381, and by the Swedish Research Council under contract 621-2009-4017.}}
\address{$^\dagger$Department of Automatic Control and ACCESS Linnaeus Center, School of Electrical Engineering, \\ KTH--Royal Institute of Technology, SE-100 44 Stockholm, Sweden. $\tt{\{nibl, crro, bo\}@kth.se}$.}
\begin{document}
%
\maketitle
\begin{abstract}
The widely used nuclear norm heuristic for rank minimization problems introduces a regularization parameter which is difficult to tune. We have recently proposed a method to approximate the regularization path, i.e., the optimal solution as a function of the parameter, which requires solving the problem only for a sparse set of points. In this paper, we extend the algorithm to provide error bounds for the singular values of the approximation. We exemplify the algorithms on large scale benchmark examples in model order reduction. Here, the order of a dynamical system is reduced by means of constrained minimization of the nuclear norm of a Hankel matrix.
\end{abstract}

\section{Introduction}

Rank minimization has important applications in e.g. signal processing, control, machine learning, system identification, and model order reduction. The matrix argument can e.g. be a covariance matrix (as in sensor array processing and multivariate statistical data analysis) or a structured matrix such as a Hankel matrix (as in system realization), \cite{Fazel:2002}. Specifically, application areas include spectrum sensing \cite{Meng:2010}, signal time delay estimation \cite{Jiang:2013}, phase retrieval of sparse signals \cite{Jaganathan:2012}, wireless network inference \cite{Papailiopoulos:2012}, channel equalization \cite{Konishi:2011}, etc.

In general, the rank minimization problem is non-convex and NP-hard \cite{Vandenberghe:1996}. However, a common convex heuristic for these problems is nuclear norm minimization. The nuclear norm $\norm{\cdot}_* = \sum_i \sigma_i(\cdot)$, i.e., the sum of the singular values, is used as a convex surrogate for the non-convex rank function; this is so because the nuclear norm can be interpreted as a convex relaxation of the rank, since it is the pointwise tightest convex function (called a \textit{convex envelope} \cite{Fazel-Hindi-Boyd-01}) to lower-bound the rank, for matrices inside a unit spectral-norm ball.

Consider a general case of minimization of the nuclear norm of a linear map subject to a quadratic constraint:

\begin{equation} \label{problem1}
\begin{aligned}
& \underset{x\in \RR^n}{\text{minimize}}
& & \norm{\Ac (x)}_* \\
& \text{subject to}
& & \norm{x-x_o}_2 \leq \lambda,
\end{aligned}
\end{equation}
where $\Ac: \RR^n \rightarrow \RR^{p\times q}$ is a linear map (for simplicity, from now on we treat the symmetric case, $p=q$), $x\in\RR^n$ is the decision variable, and $\lambda$ is the regularization parameter.

Note that the formulation in (\ref{problem1}) belongs to a subclass of regularized nuclear norm optimization problems. Other formulations include exchanging cost and constraint or the penalized version \cite{Rojas:2014}. In addition, our theory can readily be extended to weighted norms, $\norm{x}_W:=x^TWx$. Then, the quadratic constraint is equivalent to the general quadratic inequality $x^TPx + q^Tx + r \leq 0$. 

\begin{figure*}
  \includegraphics[width=\textwidth,height=3cm]{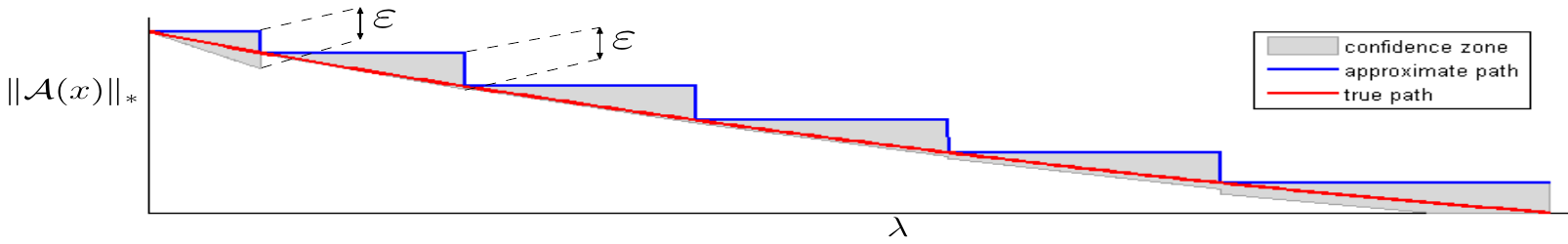}
  \caption{Illustration of regularization path algorithm proposed in \cite{Blomberg:2014}. $x$-axis: regularization parameter, $\lambda$. $y$-axis: cost of (\ref{problem1}). The true regularization path (red) is guaranteed to lie in the shaded zone. The approximate path (blue) is guaranteed to differ by at most $\varepsilon$ from the true path.}
\end{figure*}

The key issue here is that, although regularized nuclear norm minimization has been thoroughly studied, it suffers from the fact that the dependence of the solution on the regularization parameter is difficult to predict. Without, in general, \textit{a priori} knowledge on how to choose $\lambda$, we are motivated to study the so called \textit{regularization path}, i.e., the optimal solution as a function of the regularization parameter. For problem (\ref{problem1}) the regularization path is defined on the domain
\begin{equation} \label{lambda}
\lambda \in (\lambda_\text{min},\lambda_\text{max}) := (0,\norm{x_o}_2),
\end{equation}
since for $\lambda=0$ the solution to (\ref{problem1}) is known, $x^\text{opt}=x_o$, and for $\lambda \geq \norm{x_o}_2$ the constraint set is large enough to include the unconstrained minimum, $x^\text{opt}=0$. 

For practical purposes the domain of the regularization path must be discretized, which raises the question of how to choose the grid points. This is indeed an important question since problem (\ref{problem1}) can be computationally costly to solve. 

To address this problem, in \cite{Blomberg:2014}, we presented a method to choose the grid points based on a worst-case approximation error when the optimal solution for $\lambda$, $\x$, is approximated by $\xs$ for $\ls<\lambda$. The idea is visualized in Figure 1. Given the solution for some $\ls$, we increase $\lambda$ beyond $\ls$ until the worst-case approximation error reaches a pre-specified tolerance, $\varepsilon$, and then we re-evaluate (\ref{problem1}). Iteratively, starting for $\ls_0=0$, this procedure generates a set of grid points, $\ls_i,i=1,\ldots,m$, and an approximate regularization path such that the approximation error is within $\varepsilon$ for all $\lambda$.

The novelty of this paper consists of two new algorithms. The first gives a guarantee on the cost function of (\ref{problem1}). The second gives a guarantee on the singular values of $\Ac(\x)$, when $\x$ is approximated by $\xs$. Furthermore, we derive upper bounds on the number of grid points needed by the algorithms to meet a tolerance $\varepsilon$.

\parskip = 0pt
\section{Error bounds for approximation of (1)}

In this section we derive error bounds that allow us to confine the true regularization path within a certain region (the shaded area in Figure 1).

Define the singular values of $\Ac(\x)$, where $\x$ is optimal for (\ref{problem1}) for parameter value $\lambda$, as
$$
\sigma\left(\Ac\left(\x\right)\right) =: \left(\sigma_1^\lambda, \ldots, \sigma_p^\lambda\right).
$$
For further use in the below presented Algorithms 1 and 2, respectively, we derive upper bounds on the quantities:
\doubleequation[err_d,err_s]{\norm{\Ac\left(\xs\right)}_*-\norm{\Ac\left(\x\right)}_* \text{ and }}{\quad\qquad\qquad\qquad\sum\limits_i^p \left(\sigma^{\ls}_i - \sigma^\lambda_i \right)^2,}
\noindent
where $\xs$ is given. The bounds can be viewed as worst-case approximation errors in the singular values when $\x$ is approximated by $\xs$.

\subsection{Relaxation of (1) using subgradients}

We here relax problem (\ref{problem1}) using subgradients of the nuclear norm. The concept of subdifferentials (or sets of subgradients) is a generalization of the gradient that applies to functions whose gradient is undefined in some point or points, \cite{Rockafellar-70}. In the case of the nuclear norm, the subdifferential is (see e.g. \cite{Recht:2010}):
\begin{equation*}
\partial\norm{X}_* = \left\lbrace UV^T+W: U^TW = WV = 0, \norm{W}\leq 1 \right\rbrace,
\end{equation*}
where $X=U\Sigma V^T \in \RR^{p\times p}$ is a compact singular value decomposition $W \in \RR^{p\times p}$. $U^TW = W^TV = 0$ implies that $X$ and $W$ must have orthogonal row and column spaces.

Now, assume that $\xs$ solves (\ref{problem1}) for some parameter value $\lambda=\ls$. Then, since the nuclear norm is convex we can write, for any matrix $\Ac (x) \in \RR^{p\times p}$, the inequality
\begin{align*}
\norm{\Ac (x)}_* &\geq \norm{\Ac \left(\xs\right)}_* + \left\langle U_{\ls}V_{\ls}^T+W, \Ac (x)-\Ac \left( \xs \right) \right\rangle \\
&= \norm{\Ac \left(\xs\right)}_* + \Ac^*(U_{\ls}V_{\ls}^T+W)^T \left( x - \xs \right),
\end{align*}
where $U_{\ls}V_{\ls}^T+W \in \left.\partial\norm{X}_*\right|_{X=\Ac \left(\xs\right)}$, $\left\langle A,B \right\rangle=\text{Tr }B^TA$ is the standard inner product, and $\Ac^*$ is the adjoint operator of $\Ac$. For shorter notation we define
\begin{equation} \label{subgr}
\Ac^*(U_{\ls}V_{\ls}^T + W)^T=: \asW^T.
\end{equation}

To sum up, the above inequality becomes
\begin{equation} \label{ineq}
\norm{\Ac(x)}_* \geq \norm{\Ac\left(\xs\right)}_* + \asW^T\left(x-\xs\right),
\end{equation}
which implies that for $\lambda>\ls$ the optimal argument $\x$ must lie in the half-space $\left\lbrace x : \asW^T\left(x-\xs\right) \leq 0 \right\rbrace$.

Using the inequality in (\ref{ineq}) we can relax (\ref{problem1}) into
\begin{equation} \label{problem2}
\begin{aligned}
& \underset{x}{\text{min}}
& & \norm{\Ac \left(\xs\right)}_* + \asW^T\left(x-\xs\right) \\
& \text{s.t.}
& & \norm{x-x_o}_2 \leq \lambda.
\end{aligned}
\end{equation}
Problem (\ref{problem2}) is solved analytically in the following lemma:
\begin{lemma}
Problem (\ref{problem2}) has the optimal solution
\begin{equation*}
\xtwo = x_o - \frac{\lambda}{\norm{\asW}_2}\asW,
\end{equation*}
and optimal cost
\begin{equation} \label{optcostP2}
\norm{\Ac \left(\xs\right)}_* + \asW^T\left( x_o - \xs \right) - \lambda \norm{\asW}_2.
\end{equation}
\end{lemma}
\begin{proof}
At optimum the constraint is tight and the negative gradient of the cost function, $-\asW$, is proportional to the outward pointing normal of the constraint set. This gives $\xtwo$. Inserting $\xtwo$ into the cost of (\ref{problem2}) gives (\ref{optcostP2}).
\end{proof}

\subsection{Bound on cost function approximation error, (\ref{err_d})}

Using (\ref{optcostP2}) we can upper bound the approximation error in (\ref{err_d}:
\begin{theorem}
The approximation error in (\ref{err_d}) (i.e., the cost function approximation error) for any $\lambda$ is upper-bounded by the function $\dg$, as
\begin{equation} \label{dg}
\begin{aligned}
&\norm{\Ac(\xs)}_*-\norm{\Ac(\x)}_* \leq \\
& \lambda \norm{\asW}_2 - \asW^T\left( x_o - \xs \right) =: \dg.
\end{aligned}
\end{equation}
\end{theorem}
\begin{proof}
The theorem follows from the fact that, for any $\lambda$, $\norm{\Ac(\x)}_*$ is lower bounded by the optimal cost in (\ref{optcostP2}).
\end{proof}
\begin{remark} \label{rem}
In Section \ref{implem} we present a Frank-Wolfe algorithm for optimizing (\ref{dg}) over $W$. Furthermore, it can be verified that there is some $W^\text{opt}$ such that $d_{\ls}(\ls,W^\text{opt})=0$, by taking $W^\text{opt}= W^\perp$ according to (\ref{hperp}).
 \end{remark}
 \begin{remark}
In resemblance with \cite{Giesen:2012}, the function $\dg$ can be interpreted as a \textit{duality gap}, since the relaxation made in (\ref{problem2}) relates to the Frank-Wolfe algorithm \cite{Jaggi:2013} when seen as a primal-dual method.
\end{remark}

\subsection{Bound on singular value approximation error, (\ref{err_s})}

Next, we derive an upper bound on the error in (\ref{err_s}). This bound will be the minimum of two separate bounds. The first of these is as follows:
\begin{lemma}
\begin{equation} \label{firstbound}
\sum\limits_i^p \left(\sigma^{\ls}_i - \sigma^\lambda_i \right)^2 \leq \norm{\sigma^{\ls}-\norm{\Ac\left(\xs\right)}_*e_{i^\text{min}}}_2^2,
\end{equation}
where $e_i$ is the $i$'th unit vector, i.e., $e_i$ has zeros everywhere except at the $i$'th component which is one, and $i^\text{min} = \text{arg min } \sigma_i^{\ls}, i = 1, \ldots, p$.
\end{lemma}
\begin{proof}
For $\lambda>\ls$, $\norm{\Ac\left(\x\right)}_*\leq\norm{\Ac\left(\xs\right)}_*$. Hence, (\ref{firstbound}) corresponds to the maximum of $\sum\limits_i^p \left(\sigma^{\ls}_i - \sigma^\lambda_i \right)^2$ subject to $\sum\limits_i^p \sigma^\lambda_i \leq \norm{\Ac\left(\xs\right)}_*$, which is reached by making $\sigma_{i^\text{min}}$ as large as possible and $\sigma_i=0$ for $i\neq i^\text{min}$.
\end{proof}

Now, we derive a second upper bound, which is complementary to the above. To do this, consider the perturbation
\begin{equation} \label{pert}
\Ac\left(\xs\right) = \Ac\left(\x\right) + E; \quad E := \Ac\left(\xs-\x\right),
\end{equation}
which is valid since $\Ac$ is linear in $x$. Then, according to Mirsky's theorem \cite{Horn-Johnson-85} the singular values of $\Ac(\x)$ obey
$$
\sum\limits_i^p \left(\sigma^{\ls}_i - \sigma^\lambda_i \right)^2 \leq \norm{E}_F^2,
$$
where, due to equivalence of finite-dimensional norms \cite{Luenberger-69},
\begin{equation} \label{normE}
\norm{E}_F^2 = \norm{\Ac(\xs-\x)}_F^2 \leq C_\Ac\norm{\xs-\x}_2^2,
\end{equation}
for some constant $C_\Ac$ depending on $\Ac$.

Furthermore, we bound $\norm{\xs-\x}_2^2$ in Lemma \ref{lem:bound} below. For this we need the following lemma:

\begin{lemma} \label{lem:perp}
There exists a $W=W^\perp$ such that $\aperp$ (see (\ref{subgr})) is proportional to the error vector $\left(x_o-\xs\right)$, i.e.,
\begin{equation} \label{hperp}
\aperp = \gamma \left(x_o-\xs\right),
\end{equation}
for some scalar $\gamma >0$.
\end{lemma}
\begin{proof}
The proof is in the Appendix.
\end{proof}

\begin{lemma} \label{lem:bound}
\begin{equation} \label{bound}
\norm{\xs-\x}_2^2 \leq \lambda^2-(\ls)^2.
\end{equation}
\end{lemma}
\begin{proof}
Due to the existence of $\aperp$ in (\ref{hperp}), $\x$ is constrained by the convex set
$$
\Cperp := \left\lbrace x: \norm{x-x_o}_2 \leq \lambda, \aperp^T\left(x-\xs\right)\leq 0 \right\rbrace,
$$
so an upper bound is $\underset{x\in\Cperp}{\text{max}}\norm{\xs-x}_2^2$. This maximum can be solved geometrically. Since $\xs$ is inside the ball of the first constraint of $\Cperp$, this constraint has to be tight at the optima. Furthermore, with the first constraint being tight, the vectors $\left( x-x_o \right)$, $\left( \xs-x_o \right)$, and $\left( \xs-x \right)$ form a triangle, with $\norm{x-x_o}_2=\lambda$ and $\norm{\xs-x_o}_2=\ls$, so
$$
\lambda^2 = \norm{\xs-x}_2^2 + (\ls)^2 - 2\norm{\xs-x}_2\ls \cos(\frac{\pi}{2} + v)
$$
according to the law of cosines, where $v\geq 0$ is the angle between $\xs-x$ and the hyperplane
$$\left\lbrace x: \aperp^T\left(x-\xs\right) = 0 \right\rbrace.$$
This expression is maximized for $v=0$ giving the result $\norm{\xs-x^\text{opt}}_2^2 = \lambda^2-(\ls)^2$. (In fact, $v=0$ implies that the second constraint in $\Cperp$ is also tight.)
\end{proof}

Combining (\ref{firstbound}), (\ref{normE}), and (\ref{bound}), we obtain the following upper bound on the approximation error in (\ref{err_s}):
\begin{theorem}
The approximation error in (\ref{err_s}) is upper bounded by the function $\sg$:
\begin{equation} \label{sg}
\begin{aligned}
&\sum\limits_i^p \left(\sigma^{\ls}_i - \sigma^\lambda_i \right)^2 \leq \\
&\leq \min  \left\lbrace \norm{\sigma^{\ls}-\norm{\Ac\left(\xs\right)}_*e_{i^\text{min}}}_2^2 , C_\Ac\left(\lambda^2-(\ls)^2\right)\right\rbrace \\
&=: \sg.
\end{aligned}
\end{equation}
\end{theorem}
\begin{proof}
The first argument in the $\min$ is given by (\ref{firstbound}). The second is obtained by combining (\ref{normE}) and (\ref{bound}).
\end{proof}

\section{Algorithms}

\subsection{Model order reduction}

In model order reduction, and approximative filter design, the aim is to reduce a high-order model description to a low-order model while preserving the properties according to some fit criterion.

We consider a known Finite Impulse Response (FIR) model of a stable scalar discrete-time linear time-invariant dynamical system, denoted by $g_o\in\RR^n$, which is a vector containing its impulse response coefficients. Furthermore, we denote the low-order candidates by $g$, and consider the $H_2$ model fit criterion $\norm{g-g_o}_2 \leq \lambda$. Note that other criteria commonly used in model order reduction are the $H_\infty$- and Hankel norm-criteria (see \cite{Antoulas:2005} or \cite{Zhou-Doyle-Glover-96}), which are not considered here.

It can be shown \cite{Fazel-Hindi-Boyd-03} that the following Hankel matrix (here taken to be symmetric for simplicity)
\begin{equation} \label{hankel}
\Hc(g) := \begin{bmatrix}
g_1 & g_2 & \cdots & g_p \\
g_2 & g_3 & \cdots & g_{p+1} \\
\vdots & \vdots & \ddots & \vdots \\
g_p & g_{p+1} & \cdots & g_n 
\end{bmatrix},
\end{equation}
has the property that its rank is equal to the order (McMillan degree) of the dynamical system which has $g$ as impulse response. This motivates the Hankel matrix rank minimization problem to enforce a low system order.

Using the nuclear norm as surrogate for rank and the $H_2$ model fit criterion, we formulate the following special case of (\ref{problem1}):
\begin{equation} \label{probHank}
\begin{aligned}
& \underset{g}{\text{minimize}}
& & \norm{\Hc (g)}_* \\
& \text{subject to}
& & \norm{g-g_o}_2 \leq \lambda.
\end{aligned}
\end{equation}
Note that in this setting $\sigma\left(\Hc(g)\right)$ are the Hankel singular values of the system with $g$ as impulse response.

The adjoint of the Hankel operator in (\ref{hankel}), $\Hc^*(X)$, maps matrices $X\in\RR^{p\times p}$ to vectors $x\in\RR^n$, by summing the anti-diagonals of $X$, i.e.,
\begin{equation} \label{Hadj}
\Hc^*(X) = x; \quad x_k = \sum_{i+j=k+1}X_{ij}.
\end{equation}

\subsection{The algorithms}

The algorithms are outlined in Algorithm 1 and 2. The idea is to adaptively choose a set of discretization points, for which problem (\ref{problem1}) is solved. In the intermediate intervals the regularization path is approximated by the previous solution (obtained on the infimum of the current interval). The resulting approximation errors are upper bounded in (\ref{dg}) for Algorithm 1 and (\ref{sg}) for Algorithm 2. The discretization points are chosen as the values of $\lambda$ for which the upper bound reaches a pre-specified error tolerance, $\varepsilon$. This is visualised in Figure 2.

Note that in Algorithm 1, $\dg$ depends on $W$. For simplicity, we can set $W=0$, but in Section \ref{implem} we also demonstrate how to optimize $\dg$ over $W$. Also note that for a Hankel matrix the quantity $\Cc_\Ac=n$ satisfies (\ref{sg}).

\begin{algorithm}                    
\begin{algorithmic}                    
	\State \textbf{Algorithms 1 and 2.} Approximate regularization paths.
	\State \textbf{Input: } $g_o, \varepsilon$.
	\State \textbf{Output: } Approximate regularization paths such that errors (\ref{err_d}) $\leq \varepsilon$ (Alg. 1) or (\ref{err_s}) $\leq \varepsilon$ (Alg. 2) for $\lambda = [0,\lambda_\text{max}]$.
\State Initialize $i=0$. Set $\ls_0=0$.
    \While{$\ls_i \leq \lambda_\text{max}$}
    	\State Solve (\ref{problem1}) for $\lambda=\ls_i$, giving $\xsi \rightarrow \sigma^{\ls_i} = \sigma(\Ac(\xsi))$.
        \State Solve $\ls_{i+1}$ from $d_{\ls_i}(\ls_{i+1},W=0)=\varepsilon$ (Algorithm 1) or $s_{\ls_i}(\ls_{i+1})=\varepsilon$ (Algorithm 2).
        \State Accept $\xsi$ as approximate solution for $\lambda = [\ls_i,\ls_{i+1})$.
        \State Set $i = i + 1$.
    \EndWhile
\end{algorithmic}
\end{algorithm}

\subsubsection{Number of evaluations for Algorithm 1} \label{cfrp_A}

Here we bound the number of evaluations of (\ref{problem1}), i.e., the number of iterations of the above algorithm needed to guarantee the error (\ref{err_d}) within the tolerance $\varepsilon$.
\begin{theorem} \label{Malg1}
The number of evaluations of (\ref{problem1}) needed by Algorithm 1 is at most
\begin{equation}  \label{mmax1gen}
\Monegen \leq \left\lfloor \frac{2c_n\norm{g_o}_2}{\varepsilon} \right\rfloor = \mathcal{O}(\varepsilon^{-1}),
\end{equation}
in general, and if $W=0$:
\begin{equation} \label{mmax1}
\Mone \leq \left\lfloor \frac{c_n\norm{g_o}_2}{\varepsilon} \right\rfloor = \mathcal{O}(\varepsilon^{-1}),
\end{equation}
where $\varepsilon$ is the tolerance and
$$
c_n := \left( 2\sum\limits_{k=1}^{p-1}k^2 + p^2 \right)^\frac{1}{2} = \norm{\Hc^*(\mathbb{1}_{p\times p})}_2,
$$
in which $n = 2p-1$, $\mathbb{1}_{p\times p}$ is a $(p\times p)$-matrix of ones, and the adjoint of the Hankel operator is defined in (\ref{Hadj}).
\end{theorem}
\begin{proof} The proof is in the Appendix.
\end{proof}

\subsubsection{Number of evaluations for Algorithm 2} \label{svrp_A}

Now, we bound the number of evaluations of (\ref{problem1}), i.e., the number of iterations of Algorithm 2 needed to guarantee the solution within the tolerance $\varepsilon$.

\begin{theorem} \label{Malg2}
The number of evaluations of (\ref{problem1}) needed by Algorithm 2, i.e., the number of iterations, is at most
\begin{equation} \label{mmax2}
\Mtwo \leq \left\lfloor \frac{n\norm{g_o}_2^2}{\varepsilon} \right\rfloor = \mathcal{O}(\varepsilon^{-1}).
\end{equation}
\end{theorem}
\begin{proof} The proof is in the Appendix.
\end{proof}

\section{Implementation} \label{implem}
  
\begin{table*}[t]
\centering
\begin{tabular}{lrlrrr|rrr|rr}
&&&&&&\multicolumn{3}{|c|}{Algorithm 1}&\multicolumn{2}{|c}{Algorithm 2} \\
\text{benchmark} & order & $T_s$ & $n$ & \text{cpu ADMM} & $\varepsilon^\text{min}/J^\text{max}$ & $\varepsilon/J^\text{max}$ & $m$ & $M^{W=0}$ & $M$ & $m$ \\
\hline
\multirow{2}{*}{$\tt{beam.mat}$} & \multirow{2}{*}{348} & \multirow{2}{*}{1} & \multirow{2}{*}{1047} & \multirow{2}{*}{134.23} & \multirow{2}{*}{0.1233} & 0.2 & 5 & 3062 & 30 & 10\\
& & & & & & 0.3 & 3 & 4082 & 20 & 5\\
\hline
\multirow{2}{*}{$\tt{build.mat}$} & \multirow{2}{*}{48} & \multirow{2}{*}{0.025} & \multirow{2}{*}{576} & \multirow{2}{*}{317.86} & \multirow{2}{*}{0.1607} & 0.2 & 7 & 1035 & 30 & 10 \\
& & & & && 0.3& 4 & 1379 & 20 & 5 \\
\hline
\multirow{2}{*}{$\tt{eady.mat}$} & \multirow{2}{*}{598} & \multirow{2}{*}{0.1} & \multirow{2}{*}{196} &\multirow{2}{*}{7.36} & \multirow{2}{*}{0.0958} & 0.2 & 5 & 643 & 30 & 9 \\
& & & & &&0.3 & 3 & 856 & 20 & 4 \\
\hline
\multirow{2}{*}{$\tt{heat-cont.mat}$} & \multirow{2}{*}{200} & \multirow{2}{*}{0.5} & \multirow{2}{*}{139} & \multirow{2}{*}{11.78} & \multirow{2}{*}{0.7270}& 0.2 & 5 & 482 & 30 & 12 \\
& & & & && 0.3 & 3 & 642 & 20 & 7 \\
\hline
\multirow{2}{*}{$\tt{pde.mat}$} & \multirow{2}{*}{84} & \multirow{2}{*}{0.0001} & \multirow{2}{*}{242} & \multirow{2}{*}{39.80}& \multirow{2}{*}{0.1054} &  0.2 & 5 & 1193 & 30 & 7 \\
& & & & && 0.3 & 3 & 795 & 20 & 3 \\
\end{tabular}
\caption{Results of Algorithm 1 and 2. $T_s$ is sampling time in Matlab's $\tt{c2d}$, giving impulse response lengths $n$. 'cpu ADMM' is an average time in seconds with a standard laptop for solving (\ref{probHank}) using ADMM. The maximum cost $J^\text{max}:=\norm{\Hc(g_o)}_*$. $m$ is number of grid points needed, with upper bounds $M$ (for Algorithm 1 we use (\ref{mmax1})). $\varepsilon^\text{min}$ is the minimum tolerance for which  $d_{\ls}(\ls,W=0)<\varepsilon^\text{min}$ for all $\ls$. For Algorithm 2, $\varepsilon=n\norm{g_o}_2^2/\Mtwo$.}
\end{table*}

For large scale problems (\ref{problem1}) we suggest an Alternating Direction Method of Multipliers (ADMM), c.f. \cite{Boyd:2010} and \cite{Yang:2012}. We will follow the method in \cite{Liu:2013} with a modification for the $g$-update in (\ref{gupd}) below.

First, we rewrite (\ref{probHank}) as
\begin{equation}
\begin{array}{cl}
\underset{g\in\RR^n,H\in\RR^{p\times p}}{\text{minimize}} & \norm{H}_* \\
\text{subject to} & \norm{g-g_o}_2 \leq \lambda \\
 & \Hc (g) = H.
\end{array}
\end{equation}
Next, we form the following augmented Lagrangian
\begin{align*}
L_\rho &(H,g,Z) = \\
&\norm{H}_* + \text{Tr}\left(Z^T(\Hc (g) - H)\right) + \frac{\rho}{2}\norm{\Hc (g) - H}_F^2.
\end{align*}
The strategy is to update the variables as
\begin{align}
H^{k+1}&:=\underset{H}{\text{arg min }}L_\rho(H,g^k,Z^k) \label{Hupd} \\
g^{k+1}&:=\underset{\lbrace g: \norm{g-g_o}_2\leq\lambda\rbrace}{\text{arg min }}L_\rho(H^{k+1},g,Z^k) \label{gupd} \\
Z^{k+1}&:=Z^k + \rho(\Hc (g^{k+1})-H^{k+1}).
\end{align}
The variables can be initialized e.g. as $H=0,g=0,Z=0,\rho=1$. (Initialize $\rho$ if it is adaptive as in \cite{Boyd:2010}).

The $H$ update in (\ref{Hupd}) is accomplished in \cite{Liu:2013} using so called 'singular value soft-thresholding':
\begin{align*}
H^{k+1} &= \underset{H}{\text{arg min }}L_\rho(H,g^k,Z^k) \\
&= \underset{H}{\text{arg min }}\left( \norm{H}_* + \frac{\rho}{2}\norm{H - \Hc (g^k) - (1/\rho)Z^k}_F^2 \right) \\
&= \sum\limits_{i=1}^p \max \left\lbrace 0, \sigma_i - \frac{1}{\rho} \right\rbrace u_iv_i^T,
\end{align*}
where $\sigma_i,u_i,v_i$ are given by the singular value decomposition
$$
\Hc (g^k) + \frac{1}{\rho}Z^k = \sum\limits_{i=1}^p \sigma_i u_i v_i^T.
$$

The second subproblem, (\ref{gupd}), we reformulate as
\begin{align*}
\begin{array}{cl}
\underset{x}{\text{minimize}} & \frac{\rho}{2} x^T P x + q^Tx \\
\text{subject to} & \norm{x}_2 \leq \lambda,
\end{array}
\end{align*}
where $x=g-g_o$, $P=\text{diag}(\Hc^* (\mathbb{1}_{p\times p}))$ and $q = \Hc^* (Z^k + \rho\Hc (g_o) - \rho H^{k+1})$, and $\Hc^*(\cdot)$ is defined in (\ref{Hadj}). This can be solved by using the facts that the optimal point, $x^\text{opt}$, lies on the boundary of the constraint set, and in this point the negative gradient of the cost function is normal to the constraint set, i.e., it is proportional to $x^\text{opt}$. This means that
\begin{align*}
\rho P x^\text{opt} + q = -tx^\text{opt} \Leftrightarrow x^\text{opt} = -(\rho P + tI)^{-1}q
\end{align*}
where $t\geq 0$ is a scalar determined from solving $f(t) := \norm{x^\text{opt}}_2 = \lambda$ using Newton's method. This $t$ is unique since
$$
f(t) = \norm{(\rho P + tI)^{-1}q}_2 = \left(\sum\limits_{i=1}^n \frac{q_i^2}{(t+\rho P_{ii})^2}\right)^{\frac{1}{2}},
$$
is a decreasing function with $f(0)>\lambda$ and $f(\infty)=0$. The fact that $f(0)>\lambda$ is true since $x^\text{opt}(t=0)$ is the global minimum, which is located outside the constraint set. Summing up, we obtain
\begin{equation}
g^{k+1} = g_o - (\rho P + tI)^{-1}q.
\end{equation}

The stopping criterion is $\norm{r_\text{p}^{k+1}}\leq \epsilon_\text{p}^{k+1}$ and $\norm{r_\text{d}^{k+1}}\leq \epsilon_\text{d}^{k+1}$, where the primal and dual residuals ($r_\text{p}$ and $r_\text{d}$) and tolerances ($\epsilon_\text{p}$ and $\epsilon_\text{d}$) are computed from the definition in \cite[Sec.~3]{Boyd:2010} as
\begin{align*}
r_\text{p}^{k+1} &:= \Hc (g^{k+1}) - H^{k+1} \\
r_\text{d}^{k+1} &:= \rho \Hc^* (H^k-H^{k+1}) \\
\epsilon_\text{p}^{k+1} &:= p\epsilon_\text{abs} + \epsilon_\text{rel} \max \left\lbrace \norm{\Hc (g^{k+1})}_F, \norm{H^{k+1}}_F \right\rbrace \\
\epsilon_\text{d}^{k+1} &:= \sqrt{n}\epsilon_\text{abs} + \epsilon_\text{rel} \norm{\Hc^*(Z^{k+1})}_2.
\end{align*}

\subsection{Frank-Wolfe algorithm for optimizing (\ref{dg}) over $W$}

The Frank-Wolfe algorithm (or \textit{conditional gradient method}) is a simple iterative method, suggested in \cite{Frank:1956} (1956) for minimizing convex, continuously differentiable functions $f$ over compact convex sets. We here design a Frank-Wolfe algorithm for optimizing (\ref{dg}) over $W$. Our algorithm is summarized in Algorithm 3.

To solve the argument minimizations at each iteration explicitly, we note that for the constraints in $\Mc$
$$
U^TW = 0 \Leftrightarrow W = U^\perp A,
$$
and
$$
WV = U^\perp A V = 0 \Leftrightarrow AV = 0 \Leftrightarrow A = D(V^\perp)^T,
$$
for some matrices $A$ and $D$ of appropriate size. Hence,
$$
W = U^\perp D(V^\perp)^T,
$$
where $\norm{W} \leq 1 \Rightarrow \norm{D} \leq 1$. Then, in Algorithm 3, we parameterize $X=U^\perp D(V^\perp)^T$, so that $\langle X,C \rangle = \langle U^\perp D(V^\perp)^T,C \rangle = \langle D,(U^\perp)^TCV^\perp \rangle =: \langle D,\tilde{C} \rangle$, and solve
$$
\underset{\norm{D}\leq 1}{\text{arg min}} \langle D,\tilde{C} \rangle,
$$
This problem has the closed form solution $D^{\text{opt}} = U_{\tilde{C}} V^T_{\tilde{C}}$
where $\tilde{C} = U_{\tilde{C}} \Sigma_{\tilde{D}} V^T_{\tilde{C}}$ is a compact singular value decomposition. Then,
$$
X^\text{opt} = U^\perp D^{\text{opt}}(V^\perp)^T = U^\perp U_{\tilde{C}} V^T_{\tilde{C}} (V^\perp)^T.
$$

Finally, when optimizing the duality gap (\ref{dg}) over $W$ for a fixed $\lambda$, we have
$$
C = \left.\frac{\partial \dg}{\partial W}\right|_{W=W^k} = \frac{\lambda}{\norm{\hs}_2}\Hc\left(\hs\right) + \Hc\left(\xs-g_o\right).
$$

\begin{algorithm}                    
\begin{algorithmic}                    
	\State \textbf{Algorithm 3.} Frank-Wolfe for optimizing (\ref{dg}) over $W$
	\State Initialize $W^0=0$
    \For{k = 0,1,...,K}
    	\State Compute $X^\text{opt} := \underset{X\in\Mc}{\langle X, C \rangle}$; $C=\left.\frac{\partial \dg}{\partial W}\right|_{W=W^k}$
        \State Update $W^{k+1} := (1-\gamma)W^k + \gamma X^\text{opt}$, for $\gamma=\frac{2}{2+k}$
    \EndFor
\end{algorithmic}
\end{algorithm}

\begin{figure}
\centering
\includegraphics[scale=0.7]{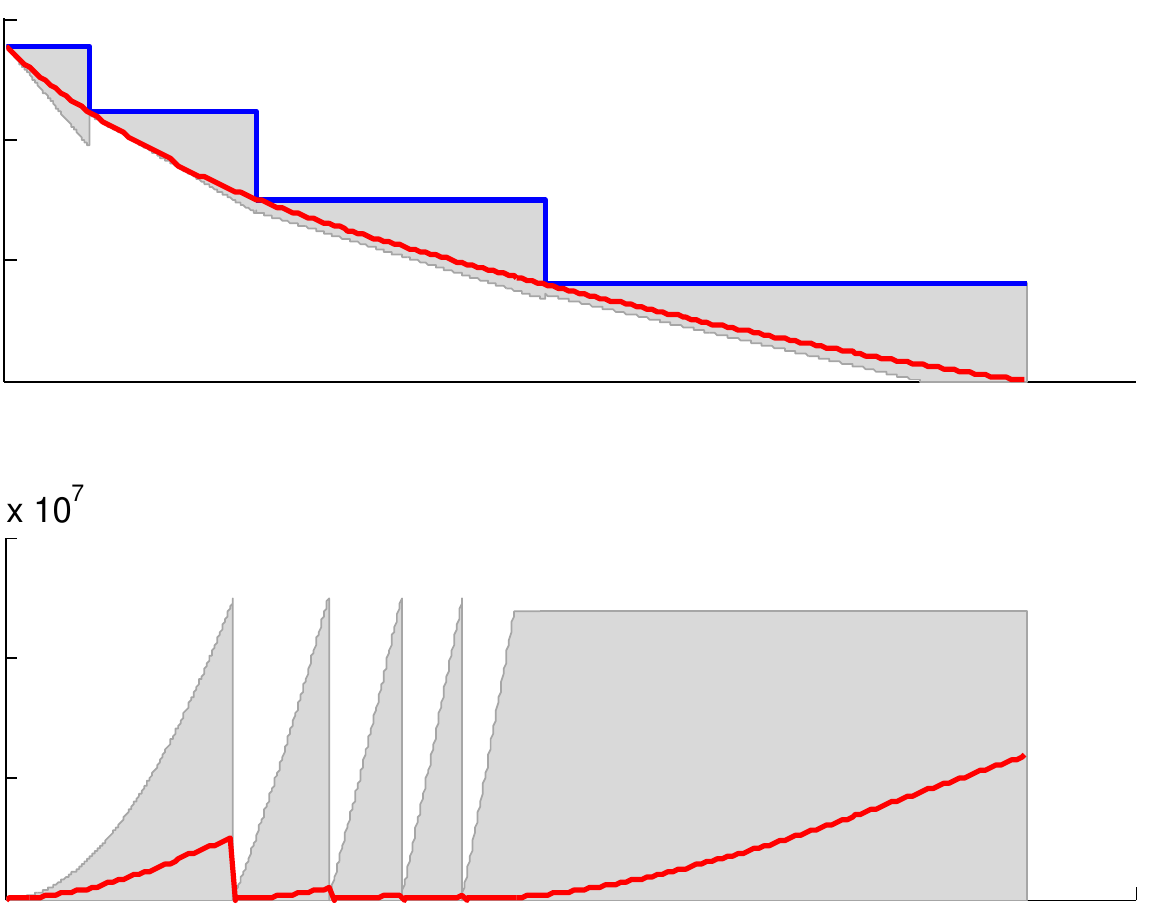}
\caption{True errors (red) and confidence zones (grey) for the model $\tt{build.mat}$. Upper: Algorithm 1 with approximate regularization path (blue) for $\varepsilon/J^\text{max} = 0.3$, where $J^\text{max}=\norm{\Hc(g_o)}_*$. Lower: Algorithm 2 for $\varepsilon=n\norm{g_o}_2^2/M$, where $M=30$.}
\end{figure}

\begin{figure}
\centering
\includegraphics[scale=0.65]{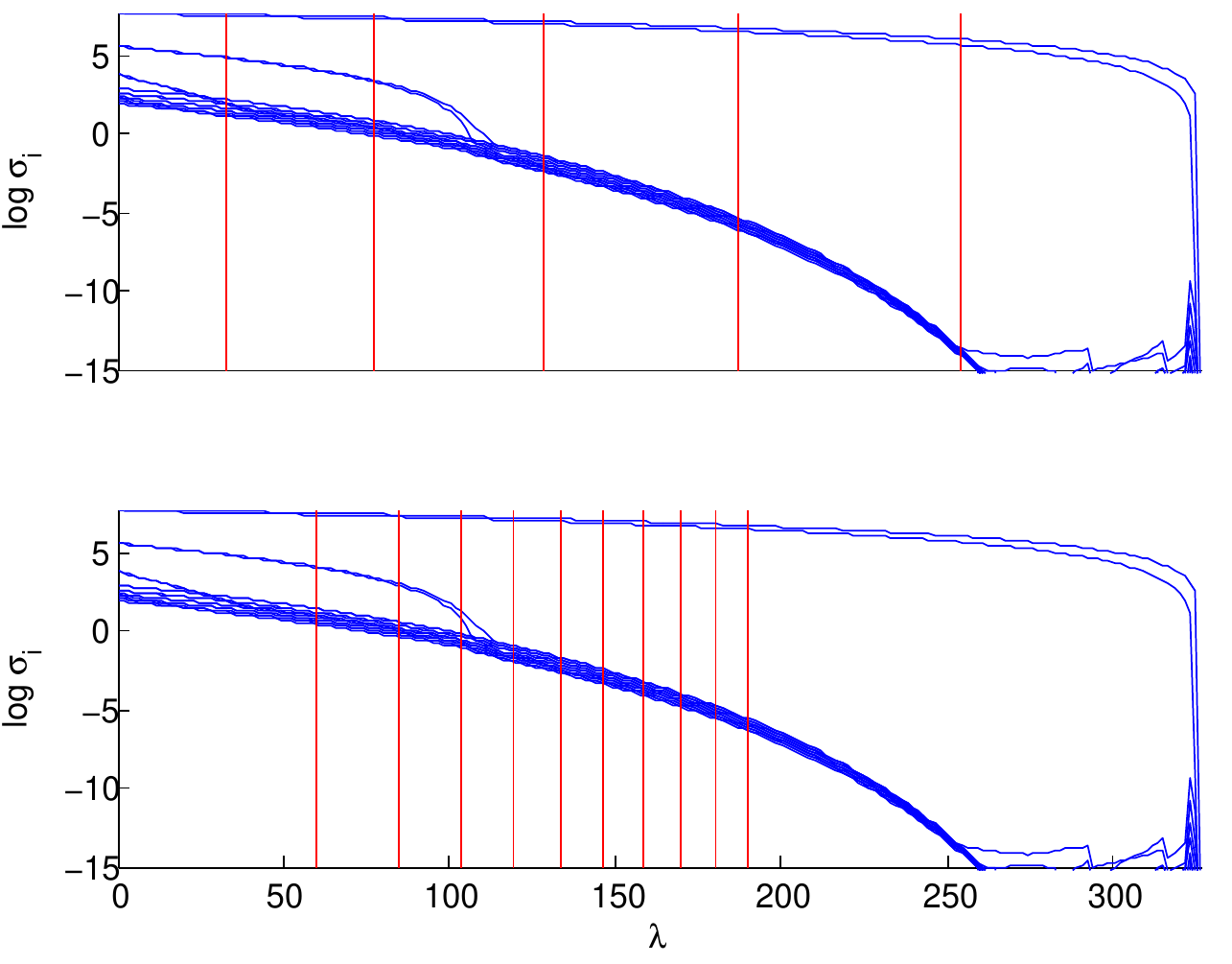}
\caption{Plot of significant singular values $\sigma_i,i=1,\ldots,17,$ for $\tt{beam.mat}$. Vertical lines indicate grid points. Upper: Algortihm 1 ($\varepsilon/J^\text{max} = 0.2$, where $J^\text{max}=\norm{\Hc(g_o)}_*$). Lower: Algorithm 2 ($\varepsilon=n\norm{g_o}_2^2/M$, where $M=40$).}
\end{figure}
\begin{figure}
\centering
\includegraphics[scale=0.65]{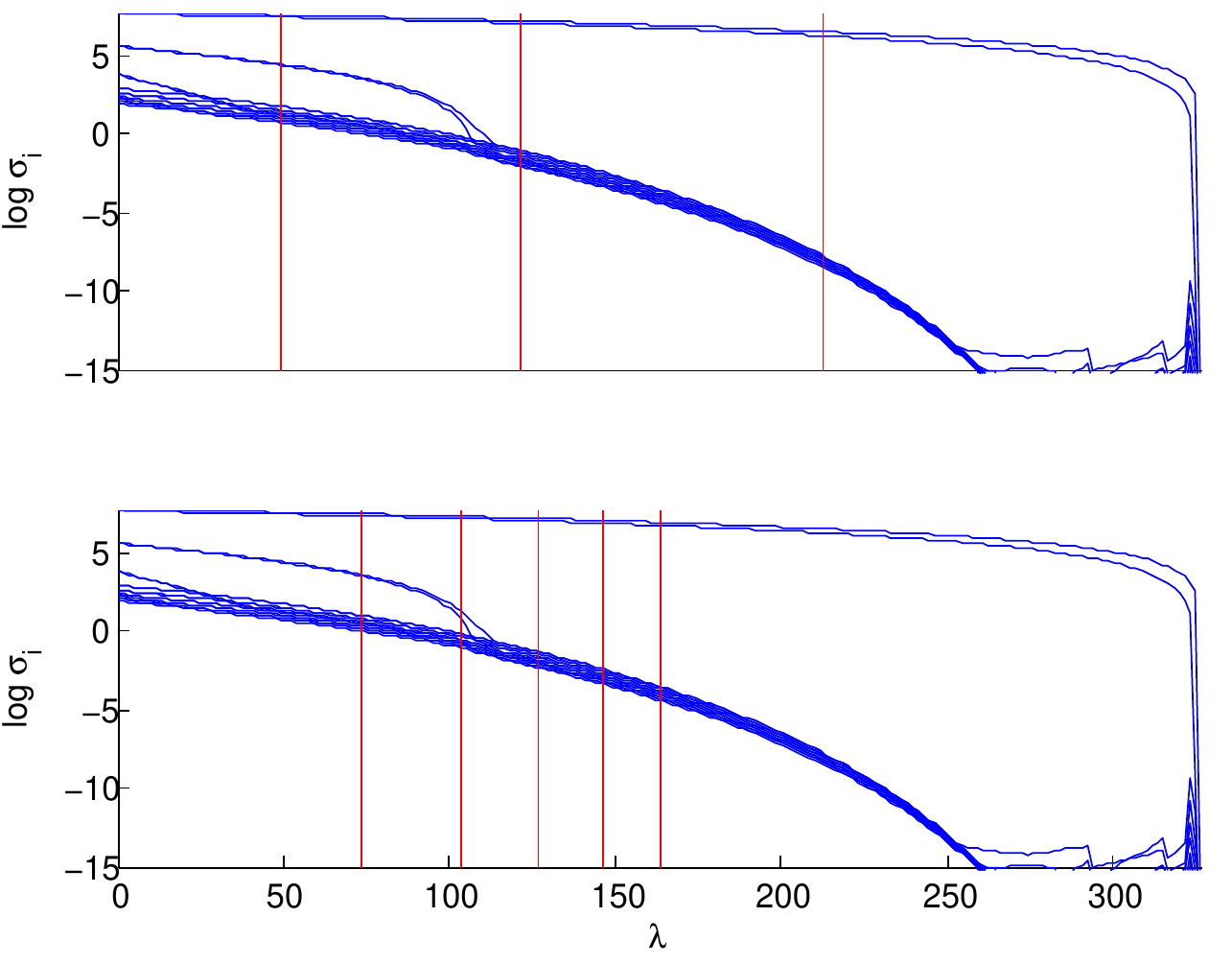}
\caption{Plot of significant singular values $\sigma_i,i=1,\ldots,17,$ for $\tt{beam.mat}$. Vertical lines indicate grid points. Upper: Algortihm 1 ($\varepsilon/J^\text{max} = 0.3$, where $J^\text{max}=\norm{\Hc(g_o)}_*$). Lower: Algorithm 2 ($\varepsilon=n\norm{g_o}_2^2/M$, where $M=30$).}
\end{figure}
\begin{figure}
\centering
\includegraphics[scale=0.65]{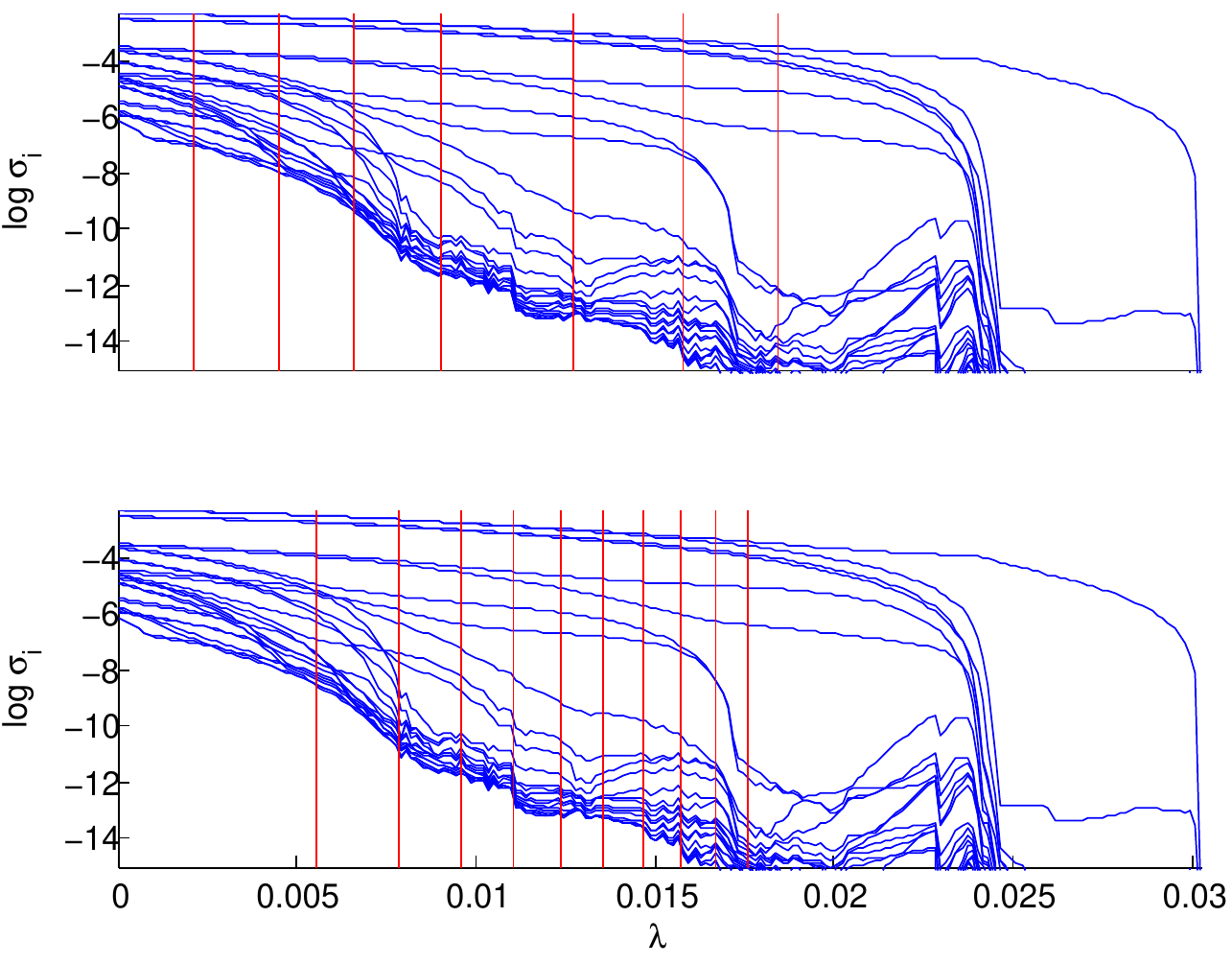}
\caption{Plot of significant singular values $\sigma_i,i=1,\ldots,26,$ for $\tt{build.mat}$. Vertical lines indicate grid points. Upper: Algortihm 1 ($\varepsilon/J^\text{max} = 0.2$, where $J^\text{max}=\norm{\Hc(g_o)}_*$). Lower: Algorithm 2 ($\varepsilon=n\norm{g_o}_2^2/M$, where $M=40$).}
\end{figure}
\begin{figure}
\centering
\includegraphics[scale=0.65]{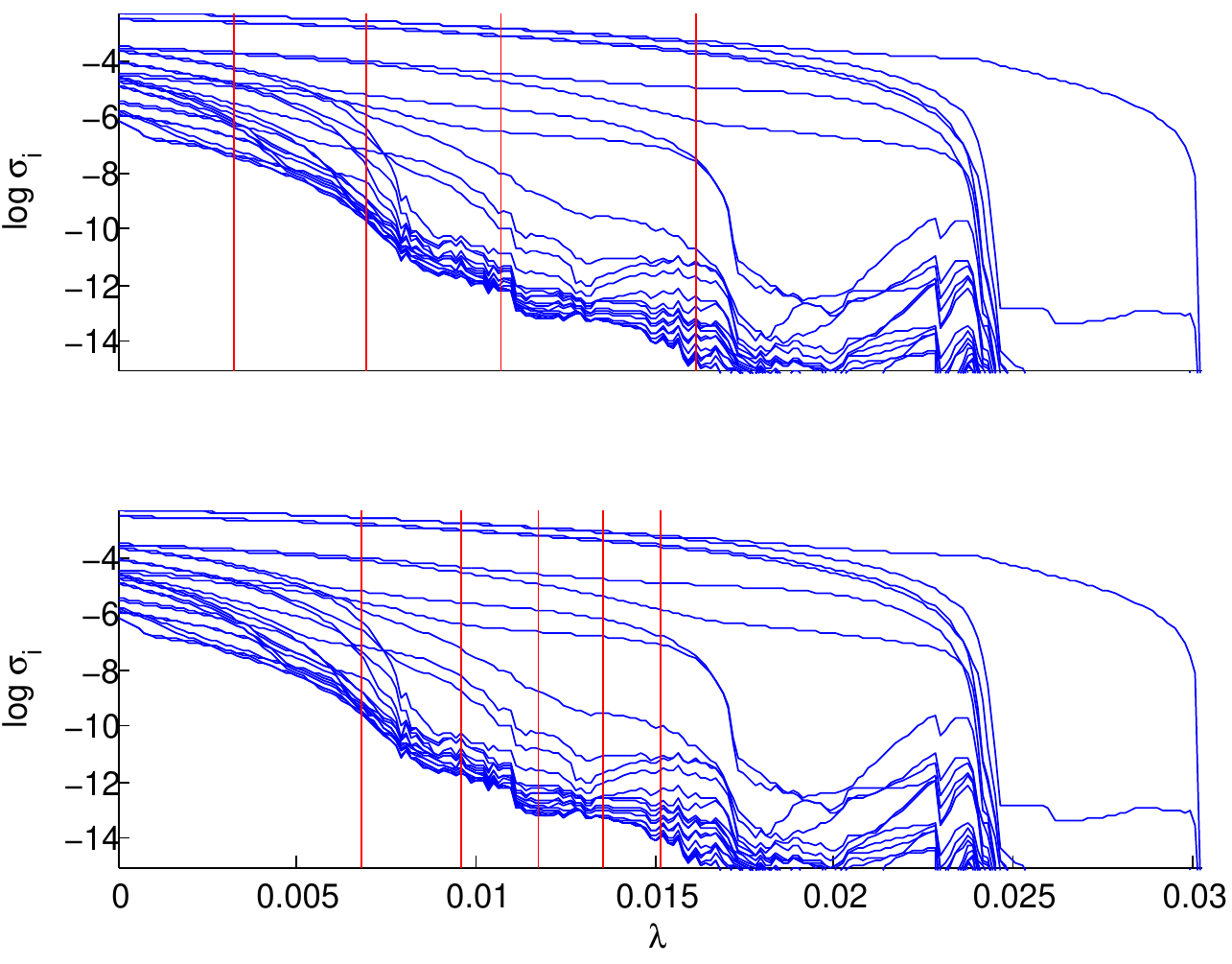}
\caption{Plot of significant singular values $\sigma_i,i=1,\ldots,26,$ for $\tt{build.mat}$. Vertical lines indicate grid points. Upper: Algortihm 1 ($\varepsilon/J^\text{max} = 0.3$, where $J^\text{max}=\norm{\Hc(g_o)}_*$). Lower: Algorithm 2 ($\varepsilon=n\norm{g_o}_2^2/M$, where $M=30$).}
\end{figure}
\begin{figure}
\centering
\includegraphics[scale=0.65]{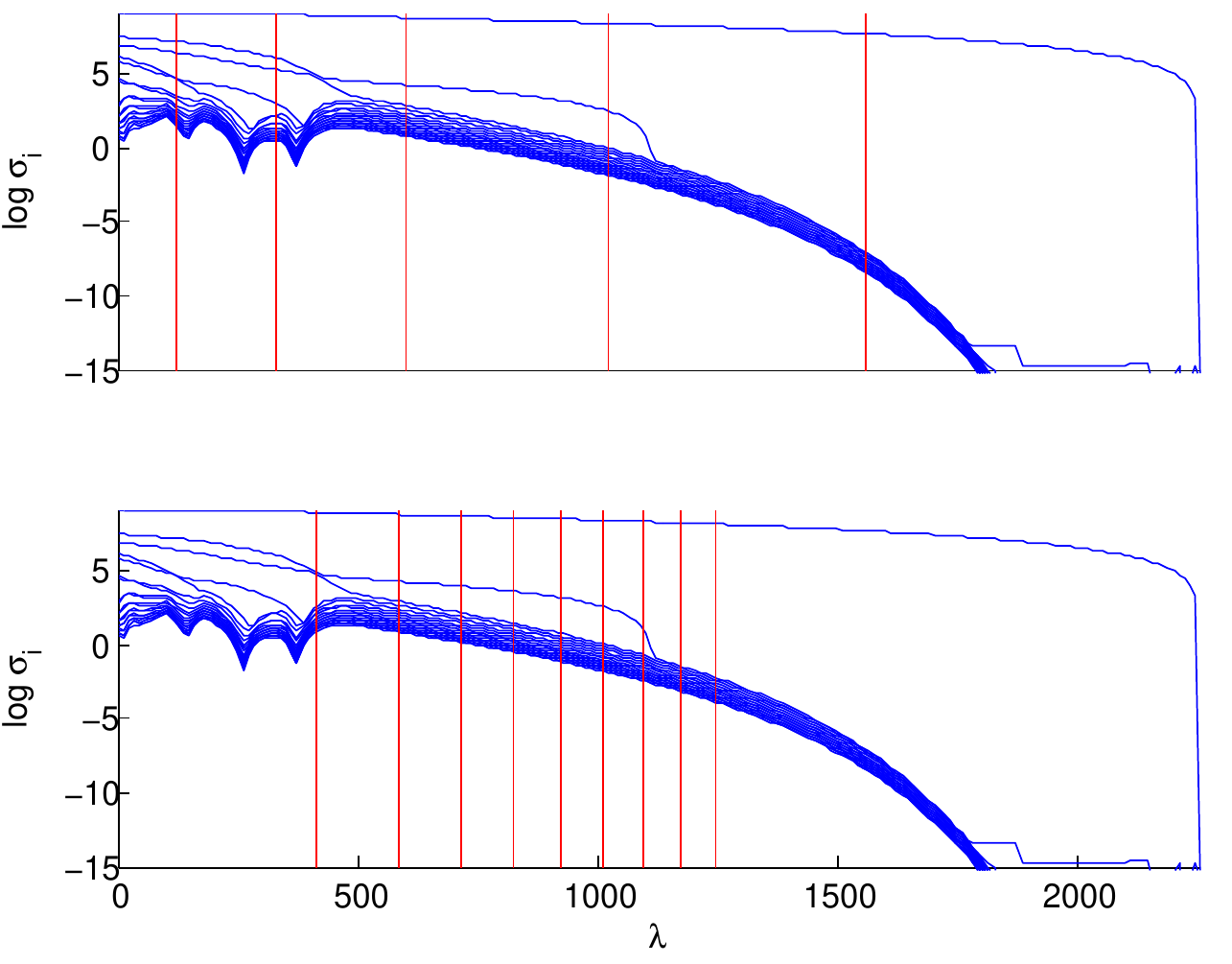}
\caption{Plot of significant singular values $\sigma_i,i=1,\ldots,20,$ for $\tt{eady.mat}$. Vertical lines indicate grid points. Upper: Algortihm 1 ($\varepsilon/J^\text{max} = 0.2$, where $J^\text{max}=\norm{\Hc(g_o)}_*$). Lower: Algorithm 2 ($\varepsilon=n\norm{g_o}_2^2/M$, where $M=40$).}
\end{figure}
\begin{figure}
\centering
\includegraphics[scale=0.65]{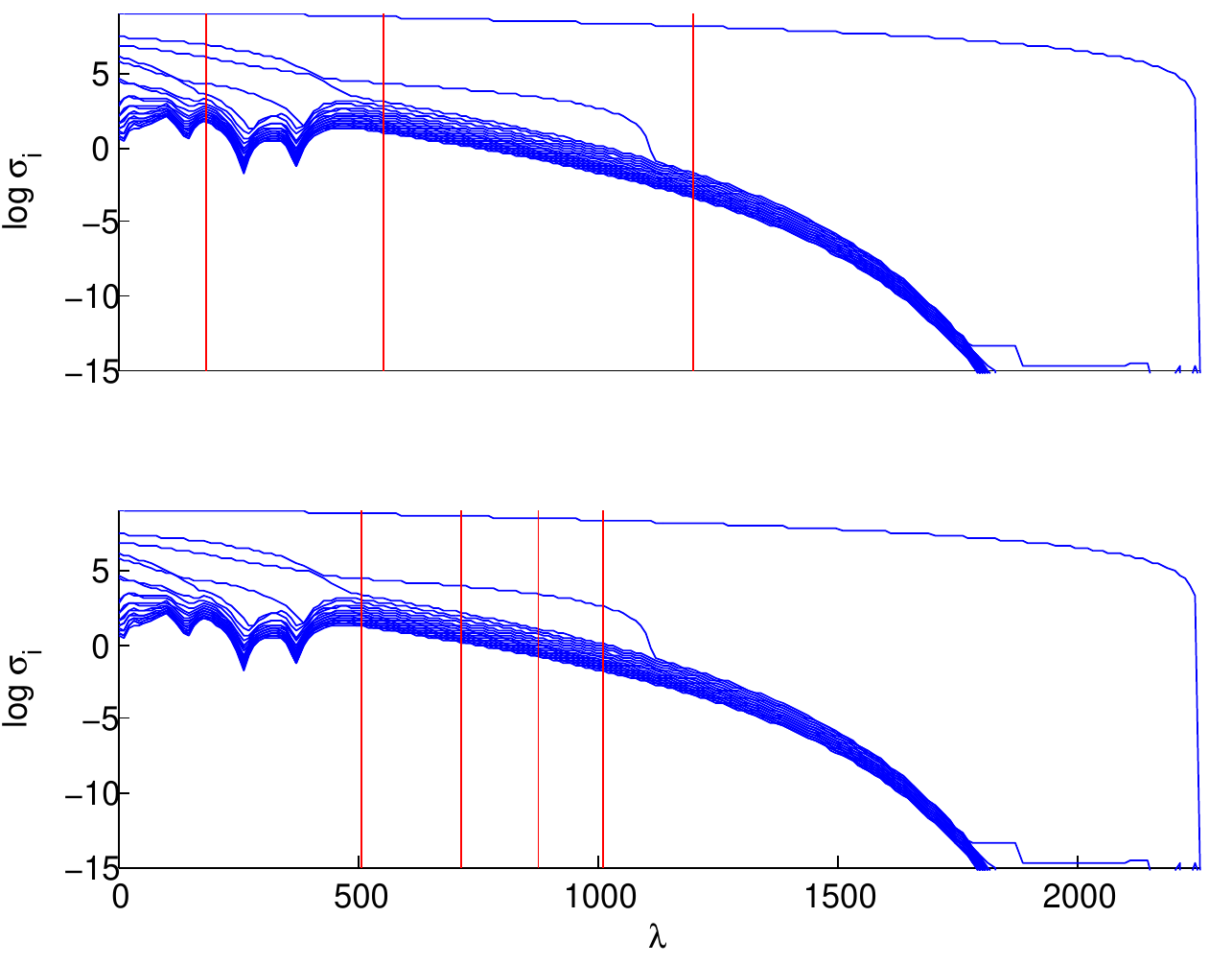}
\caption{Plot of significant singular values $\sigma_i,i=1,\ldots,20,$ for $\tt{eady.mat}$. Vertical lines indicate grid points. Upper: Algortihm 1 ($\varepsilon/J^\text{max} = 0.3$, where $J^\text{max}=\norm{\Hc(g_o)}_*$). Lower: Algorithm 2 ($\varepsilon=n\norm{g_o}_2^2/M$, where $M=30$).}
\end{figure}
\begin{figure}
\centering
\includegraphics[scale=0.65]{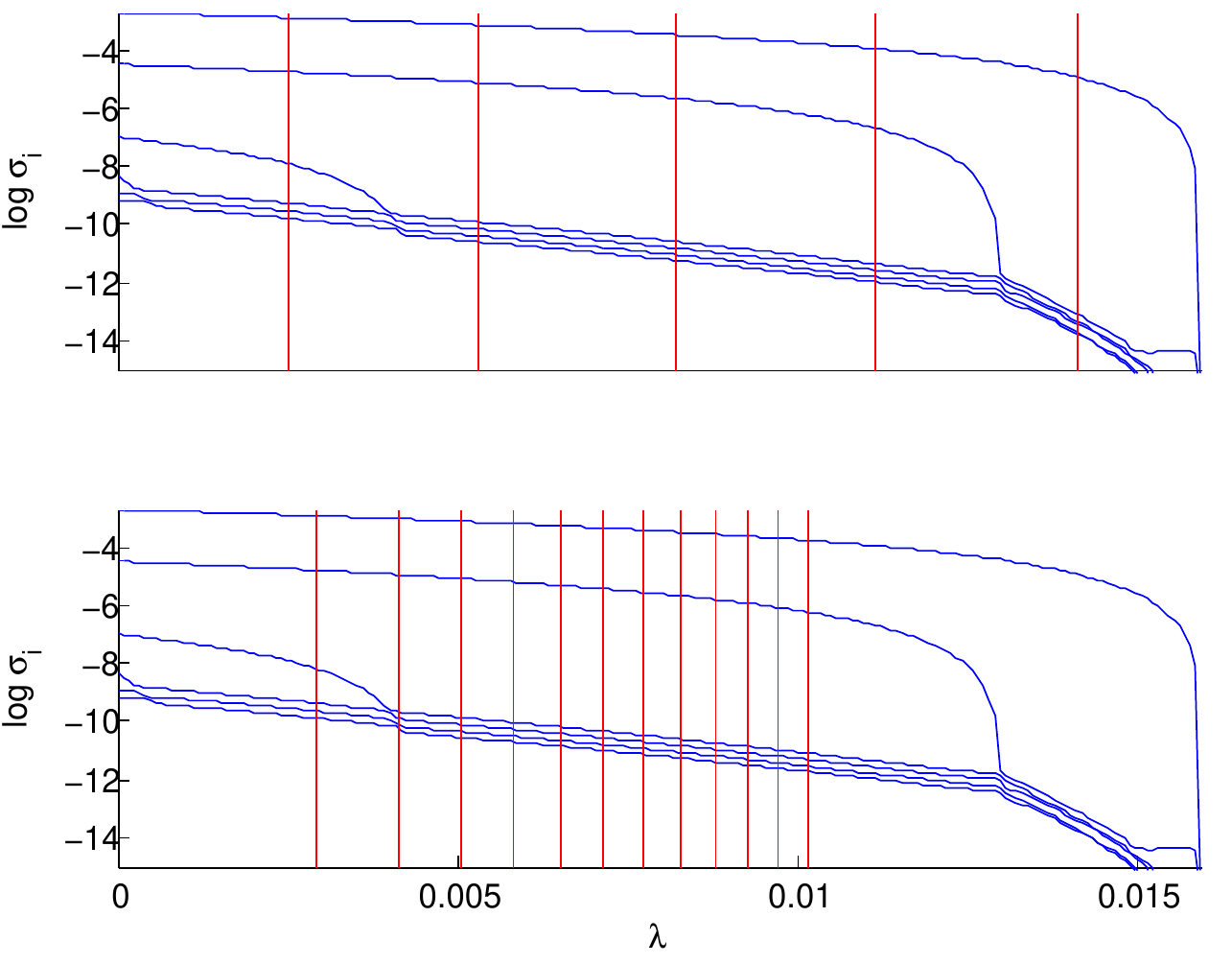}
\caption{Plot of significant singular values $\sigma_i,i=1,\ldots,6,$ for $\tt{heat-cont.mat}$. Vertical lines indicate grid points. Upper: Algortihm 1 ($\varepsilon/J^\text{max} = 0.2$, where $J^\text{max}=\norm{\Hc(g_o)}_*$). Lower: Algorithm 2 ($\varepsilon=n\norm{g_o}_2^2/M$, where $M=40$).}
\end{figure}
\begin{figure}
\centering
\includegraphics[scale=0.65]{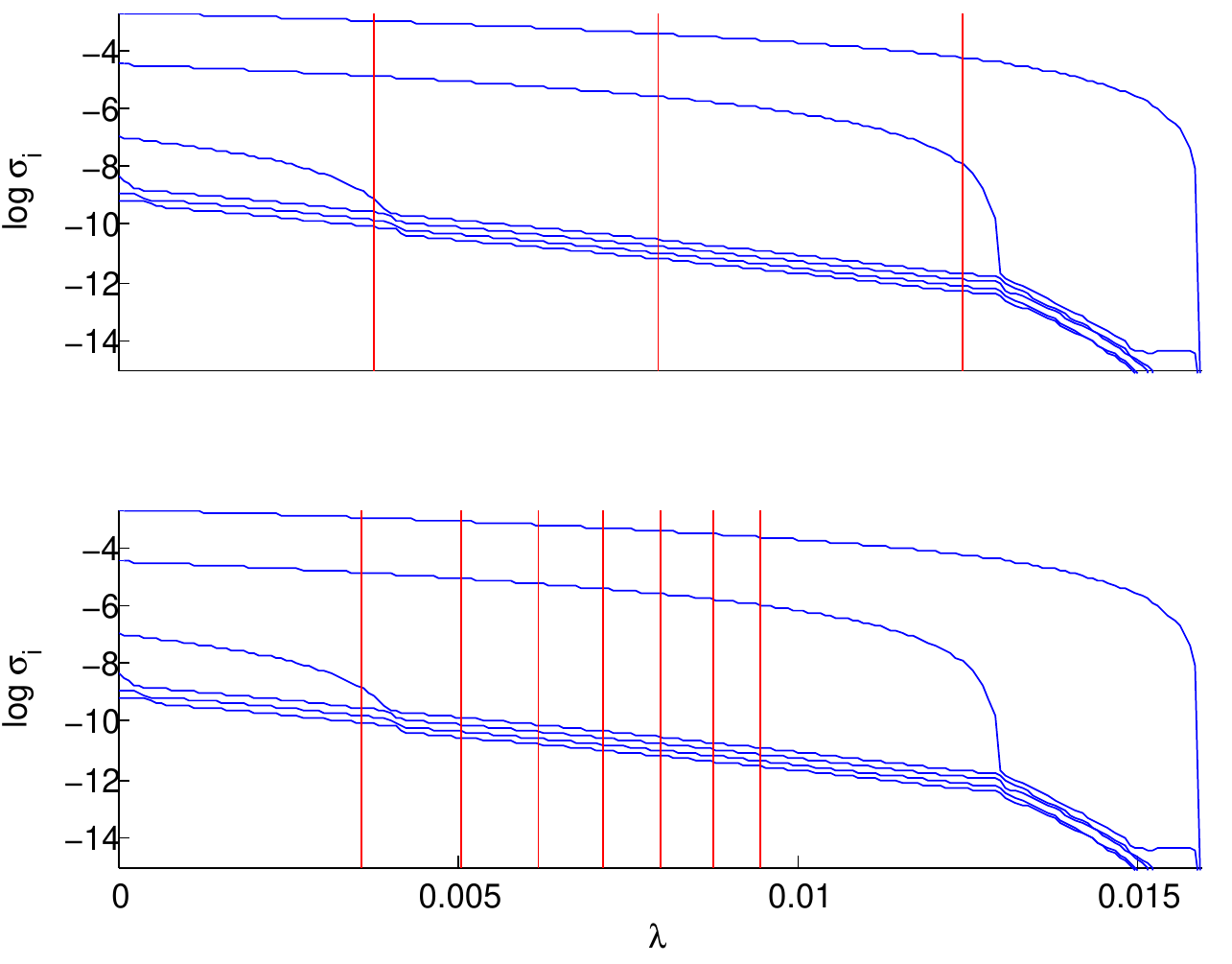}
\caption{Plot of significant singular values $\sigma_i,i=1,\ldots,6,$ for $\tt{heat-cont.mat}$. Vertical lines indicate grid points. Upper: Algortihm 1 ($\varepsilon/J^\text{max} = 0.3$, where $J^\text{max}=\norm{\Hc(g_o)}_*$). Lower: Algorithm 2 ($\varepsilon=n\norm{g_o}_2^2/M$, where $M=30$).}
\end{figure}
\begin{figure}
\centering
\includegraphics[scale=0.65]{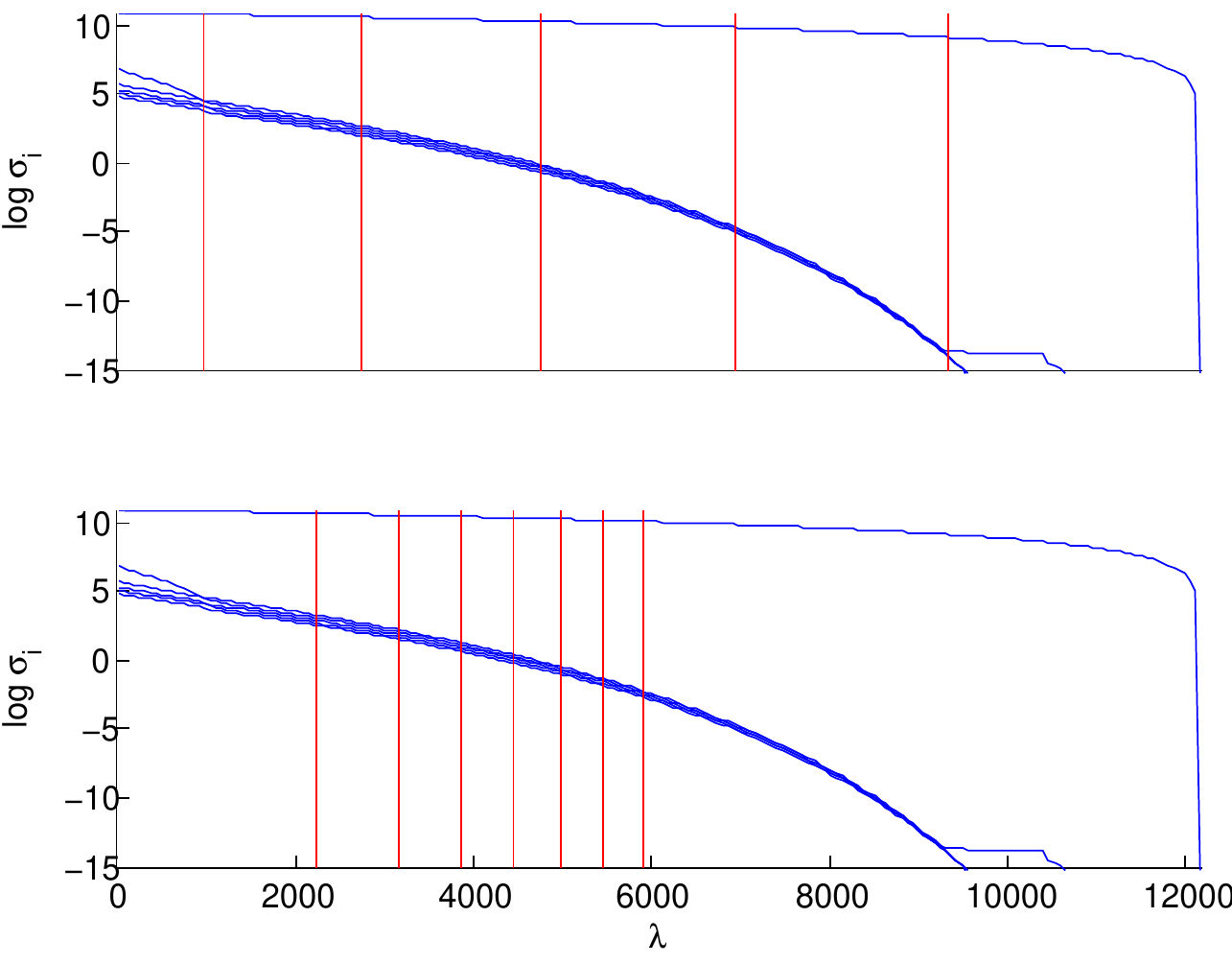}
\caption{Plot of significant singular values $\sigma_i,i=1,\ldots,6,$ for $\tt{pde.mat}$. Vertical lines indicate grid points. Upper: Algortihm 1 ($\varepsilon/J^\text{max} = 0.2$, where $J^\text{max}=\norm{\Hc(g_o)}_*$). Lower: Algorithm 2 ($\varepsilon=n\norm{g_o}_2^2/M$, where $M=40$).}
\end{figure}
\begin{figure}
\centering
\includegraphics[scale=0.65]{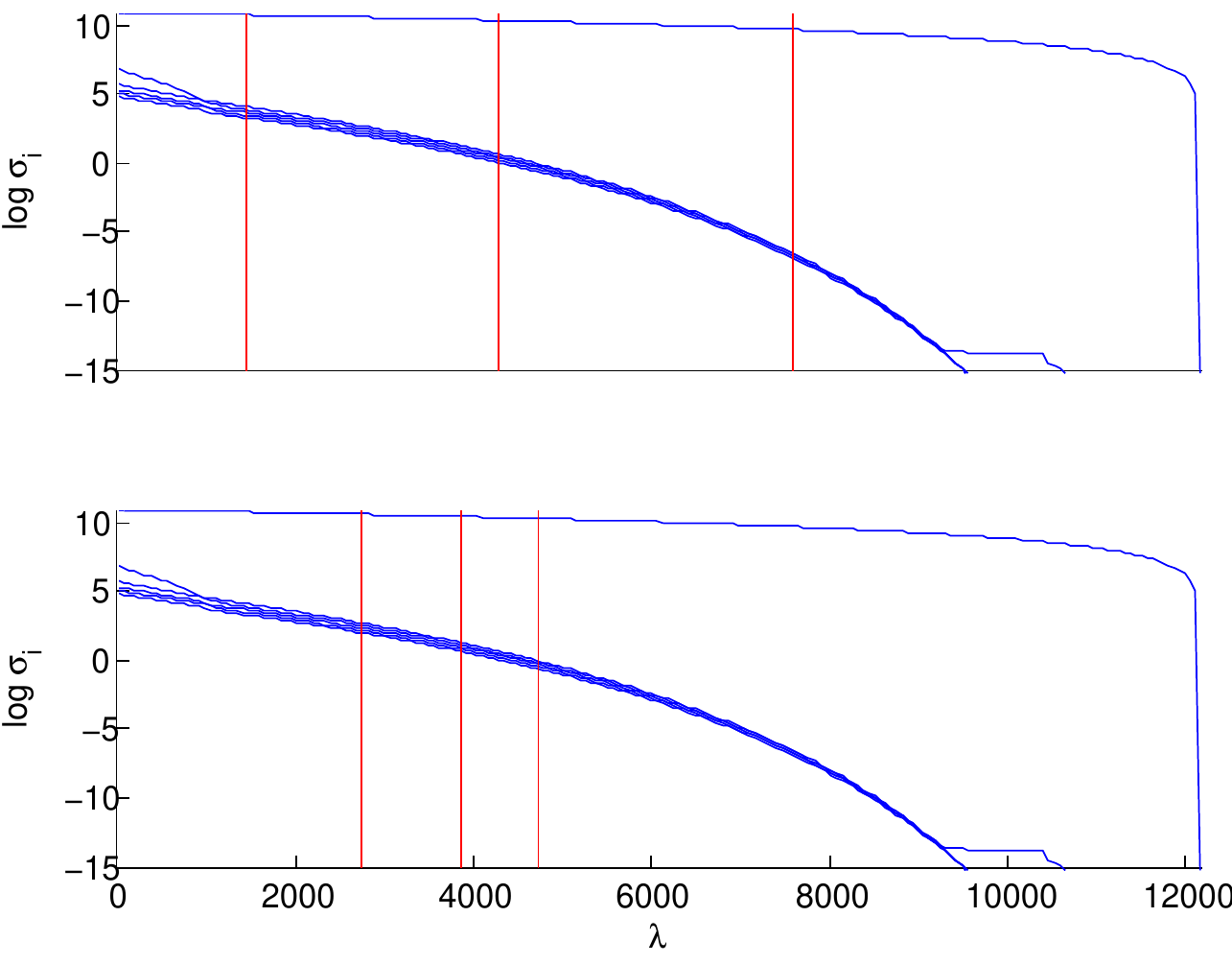}
\caption{Plot of significant singular values $\sigma_i,i=1,\ldots,6,$ for $\tt{pde.mat}$. Vertical lines indicate grid points. Upper: Algortihm 1 ($\varepsilon/J^\text{max} = 0.3$, where $J^\text{max}=\norm{\Hc(g_o)}_*$). Lower: Algorithm 2 ($\varepsilon=n\norm{g_o}_2^2/M$, where $M=30$).}
\end{figure}

\section{Results}

Algorithm 1 and 2 are implemented on single-input-single-output model order reduction benchmarks\footnote{The model reduction benchmarks are avaliable at slicot.org/20-site/126-benchmark-examples-for-model-reduction}. The continuous-time impulse response sequences are discretized using the Matlab command $\tt{c2d}$ with sampling times listed in Table 1. We here set $W=0$, but design a Frank-Wolfe algorithm for optimizing (\ref{dg}) over $W$ in Section \ref{implem}. The tolerances $\varepsilon$ are chosen as a fraction the maximum possible cost of (\ref{probHank}) in Algorithm 1, and according to $\varepsilon=n\norm{g_o}_2^2/\Mtwo$ for Algorithm 2.

For large scale problems (\ref{problem1}) we suggest an Alternating Direction Method of Multipliers (ADMM), c.f. \cite{Boyd:2010}, \cite{Yang:2012}. Our method is similar to \cite{Liu:2013} and provided in Section \ref{implem}.

The results are summarized in Table 1. We observe that, for Algorithm 1, the bound $\Monegen$ is very loose. We also see the smallest possible $\varepsilon$ such that $d_{\ls}(\ls,W=0)<\varepsilon$ for any $\ls\in(\lambda_\text{min},\lambda_\text{max})$. For the system $\tt{heat-cont.mat}$ this extreme value is high, but we observe that for most part of the regularization path $d_{\ls}(\ls,W=0)$ is very low; it only increases very close to $\lambda_\text{max}$. For smaller values of $\varepsilon$ in Algorithm 1, we may optimize over $W$. Then, it is possible to use arbitrarily small $\varepsilon$, since $d_{\ls}(\ls,W^\perp)=0$ (see Remark \ref{rem}).

In Figure 2 we illustrate the ideas in Algorithm 1 and 2.

In Figures 3-12 we visualize the grid points for Algorithm 1 and 2, respectively, when applied to the model the different benchmark models. For Algorithm 1, we generally see that the grid points are slightly more dense in for smaller values of $\lambda$. For Algorithm 2, the grid points are also more dense in the first part of the regularization path; it stops gridding when the first argument in (\ref{sg}) becomes active. Thus, the algorithms are suitable in cases where the singular value drops happen for low values of $\lambda$, which is a general observation we have made when studying these benchmark examples.

\section{Conclusion}

We have proposed a method to approximate the regularization path of a quadratically constrained nuclear norm minimization problem, with guarantees on the singular values of the matrix argument. The algorithms solve the problem for a finite, explicitly upper-bounded, number of values of the regularization parameter. We have also provided details regarding efficient implementation of the algorithms.

The results show that the algorithms generate grid points that are suitable for tracking changes in the singular values.

\section{Appendix}
\begin{proof}[Proof of Lemma \ref{lem:perp}]
The existence of vectors proportional to $x_o-\xs$ can be proved using the optimality conditions of (\ref{problem1}), which imply that the minimizer of (\ref{problem1}), $\x$, minimizes the Lagrangian $L(g,z)$ for some Lagrange multiplier $z\geq 0$ \cite[pp.~217]{Luenberger-69}, i.e.,
$$
\begin{aligned}
0 &\in \partial_x \left. L(x,z) \right|_{x=\xs} \\
&= \partial_x \left.\norm{\Ac (x)}_*\right|_{x=\xs} + \partial_x z \left.\norm{x-x_o}_2\right|_{x=\xs}.
\end{aligned}
$$
The second term is computed explicitly as
$$
\left.\partial_x z \norm{x-x_o}_2\right|_{x=\xs} = \frac{z}{\norm{\xs-x_o}_2}\left(\xs-x_o\right),
$$
so that
$$
\frac{z}{\norm{\xs-x_o}_2}\left(x_o-\xs\right) \in \left.\partial_x \norm{\Ac (x)}_*\right|_{x=\xs},
$$
showing that there is a subgradient proportional to $x_o-\xs$.
\end{proof}

\begin{proof}[Proof of Theorem \ref{Malg1}]
Consider the choice of $\hsW = \hperp$ in (\ref{hperp}). Then, for $\lambda>\ls$,
$$
d_{\ls}\left(\lambda,W^\text{opt}\right) \leq d_{\ls}\left(\lambda,W^\perp\right) = \norm{\hperp}_2\left(\lambda-\ls\right),
$$
in which we will now bound $\norm{\hperp}_2$. To do this, we can bound
\begin{equation} \label{proofineq}
\begin{aligned}
&\norm{\hsW}_2 = \left(\sum\limits_{k=1}^n \left( \text{Tr } (UV^T+W)^TH_k\right)^2\right)^\frac{1}{2} = \\
&\left(\sum\limits_{k=1}^n \left( \text{Tr } VU^TH_k + \text{Tr } W^TH_k \right)^2\right)^\frac{1}{2} \leq \left(\sum\limits_{k=1}^n \left( 2\norm{H_k}_*\right)^2\right)^\frac{1}{2} \\
&= 2 \left(\sum\limits_{k=1}^p \norm{I_k}_*^2 + \sum\limits_{k=p+1}^n \norm{I_{n-k+1}}_*^2\right)^\frac{1}{2} = 2c_n,
\end{aligned}
\end{equation}
where we have used in the first inequality twice the characterization of the nuclear norm as
$$
\norm{X}_* = \text{sup} \left\lbrace\text{tr} (Y^TX) : \norm{Y}\leq 1\right\rbrace
$$
with $Y=UV^T$ such that $\norm{UV^T}=1$, and $Y=W$ such that $\norm{W}\leq 1$, and in the second last equality the unitary invariance of the nuclear norm.

Using the bound on $\norm{\hsW}_2$ in $\norm{\hperp}_2\left(\lambda-\ls\right)$ the sub-intervals of the algorithm are of length at most
$$
\ls_{i+1}-\ls_i = \frac{\varepsilon}{2c_n}.
$$
Hence, we need at least
$$
\left\lfloor \frac{2c_n(\lambda_\text{max}-\lambda_\text{min})}{\varepsilon} \right\rfloor
$$
evaluations of (\ref{problem1}). Since $(\lambda_\text{min},\lambda_\text{max})=(0,\norm{g_o}_2)$ according to (\ref{lambda}), we obtain (\ref{mmax1gen}). If we set $W=0$ in (\ref{proofineq}) we obtain (\ref{mmax1}).
\end{proof}

\begin{proof}[Proof of Theorem \ref{Malg2}]
If the first argument in (\ref{sg}) is ignored Algorithm 2 will use a greater or equal number of grid points. In this case, Algorithm 2 evaluates (\ref{problem1}) for $\lambda=\ls_i, i=1,\ldots,\Mtwo$, obtained from $s_{\ls_i}(\ls_{i+1})=\varepsilon$. As in (\ref{sg}), $\sg = n\left(\lambda^2-(\ls)^2\right)$, which gives the relation
$$
\ls_{i+1} = \left(\frac{1}{n} \varepsilon + (\ls_{i})^2\right)^{1/2}.
$$
By induction in $i$, with $\ls_0=0$, this becomes
$$
\ls_{i+1} = \left(\frac{i}{n} \varepsilon \right)^{1/2}.
$$
Given that $\lambda_\text{max} = \norm{g_o}_2$, as in (\ref{lambda}), the largest integer $\Mtwo$ such that $\left(\frac{\Mtwo}{n} \varepsilon \right)^{1/2}\leq \norm{g_o}_2$ obeys $\Mtwo = \left\lfloor \frac{n\norm{g_o}_2^2}{\varepsilon} \right\rfloor$. Since this is a worst-case scenario we obtain the upper bound in (\ref{mmax2}).
\end{proof}

\bibliographystyle{IEEEbib}
\bibliography{niclas_SPW}

\end{document}

%% file: shortcutsNiclas.tex
\newcommand{\RR}{\mathbb R}

\newcommand{\bbmat}{\begin{bmatrix}}
\newcommand{\ebmat}{\end{bmatrix}}

\newcommand{\norm}[1]{\left\lVert #1 \right\rVert}

\newtheorem{theorem}{Theorem}[section]

\newtheorem{lemma}{Lemma}[section]

\newtheorem{remark}{Remark}[section]

\newcommand{\Ac}{{\mathcal A}}

\newcommand{\Cc}{{\mathcal C}}

\newcommand{\Hc}{{\mathcal H}}

\newcommand{\Mc}{{\mathcal M}}
